\documentclass[pra,twocolumn,showpacs,floatfix]{revtex4}
\usepackage{amssymb}
\usepackage{graphicx}
\usepackage{amsmath}
\usepackage{amsfonts}
\usepackage{tabularx,ragged2e,booktabs}
\usepackage{verbatim} % block comment
\usepackage[format=plain,justification=centerlast]{caption}

\newenvironment{proof}[1][Proof]{\noindent\textbf{#1.} }{\ \rule{0.5em}{0.5em}}

\newtheorem{mylemma}{Lemma}
\newtheorem{mydefinition}{Definition}
\newtheorem{mytheorem}{Theorem}

\newtheorem{mycorollary}{Corollary}

\newcommand{\tr}{\operatorname{Tr}}

\newcommand{\beq}{\begin{equation}}
\newcommand{\eeq}{\end{equation}}
\newcommand{\baq}{\begin{eqnarray}}
\newcommand{\eaq}{\end{eqnarray}}
\newcommand{\unit}{1\!\!1}
\usepackage[margin=0.5cm]{caption}

\begin{document}

\title{A Method for Generating All Uniform $\pi$-Pulse Sequences Used in Deterministic Dynamical Decoupling}
\author{Haoyu Qi}
\affiliation{Hearne Institute for Theoretical Physics and  Department of Physics \& Astronomy,Louisiana State University, Baton Rouge, Louisiana 70803, USA}
\author{Jonathan P. Dowling}
\affiliation{Hearne Institute for Theoretical Physics and Department of Physics \& Astronomy,Louisiana State University, Baton Rouge, Louisiana 70803, USA}

\begin{abstract}
 Dynamical decoupling has been actively investigated since Viola first suggested using a pulse sequence to protect a qubit from decoherence. Since then, many schemes of dynamical decoupling have been proposed to achieve high-order suppression, both analytically and numerically. However, hitherto, there has not been a systematic framework to understand all existing uniform $\pi$-pulse dynamical decoupling schemes. In this report, we use the projection pulse sequences as basic building blocks and concatenation as a way to combine them. We derived a new concatenated-projection dynamical decoupling (CPDD), a framework in which we can systematically construct pulse sequences to achieve arbitrary high suppression order. All previously known uniform dynamical decoupling sequences using $\pi$ pulse can be fit into this framework. Understanding uniform dynamical decoupling as successive projections on the Hamiltonian will also give insights on how to invent new ways to construct better pulse sequences.
\end{abstract}

\maketitle

\section{Introduction}

\label{sec:intro}
One of the major difficulties in the realization of quantum computing and quantum information processing is protecting the quantum state from decoherence. Quantum error correction protocols were developed to meet this challenge \cite{Shor:1995,Steane:1996,Knill:1997}. For a review and recent progress of this area, see references \cite{Zoller:2005,Steane:1998}. 

While quantum error correction can be regarded as a form of closed-loop control, dynamical decoupling has been proposed as a way to counteract the interaction between a quantum system and the environment by an open-loop control field. The idea of using pulse sequences, to protect nuclear spins from classical decoherence, dates back to 1950, when the spin-echo method was found \cite{Hahm:1950}. Since then, many pulse methods have been developed in nuclear magnetic resonance spectroscopy \cite{NMR:book}. In 1998, it was first pointed out that a similar technique, periodical dynamical decoupling (PDD), can be applied to open quantum systems \cite{Viola:98}. By using a control field, with duration shorter than the time scale of the environment, dynamical decoupling can suppress the interaction between a qubit and the bath \cite{Viola:98}. Important theoretical progress was made by understanding  the effect of  dynamical decoupling as the result of a symmetrizing procedure over the group composed of all independent $\pi$ pulses \cite{Viola:99,Zanardi:99}. However, the finite switching time in real experimental conditions makes the symmetrizing imperfect. Therefore, concatenated dynamical decoupling (CDD) was proposed to eliminate higher-order errors in the interaction Hamiltonian \cite{KhodjastehLidar:04,KhodjastehLidar:07}. 

In realistic experiments, the imperfection of pulses such as the rotation angle error and the finite width, must be taken into account. CDD is preferred over PDD in this case, due to the fact that concatenation not only suppresses the interaction with the environment but also the pulse errors to higher order \cite{KhodjastehLidar:04,KhodjastehLidar:07}.

Meanwhile, by optimizing the pulse interval to suppress the low-frequency region of the noise spectrum, Uhrig dynamical decoupling (UDD) can achieve the same suppression order with  fewer pulses, compared with CDD \cite{Uhrig:07,Uhrig:08}. Although it has superior performance, UDD is very sensitive to pulse errors, due to the fact that it only uses single-axis rotations. However, to protect unknown states, uniform DD with multi-axis rotations can compensate errors due to its symmetric structure \cite{Xiao:11,Lange:2010,Ryan:2010,Souza:2012}. \\
\indent Another development in dynamical decoupling is to use random pulses, instead of deterministic schemes, for sufficiently long sequences \cite{Viola:05,Viola:08}. Instead of trying to suppress the interaction to arbitrarily high orders, random dynamical decoupling schemes improves the time dependence of the error accumulation from  quadratic to linear \cite{Viola:05,Viola:08}.\\ 
% %
\indent In this paper we only consider deterministic uniform dynamical decoupling for two reasons: (1) it is easy to implement in experiment, and (2) it is robust against pulse error compared to UDD. Besides analytical calculations, recently many other DD sequences have been found using genetic algorithms to optimize the suppression order \cite{Quiroz:13}, some of which even achieve same suppression order with fewer pulses than CDD. However, a unified understanding of all known DD schemes has not been developed until now.\\ 
% %
\indent In this work, we propose concatenated-projection dynamical decoupling (CPDD) to unify all known uniform DD schemes. Our framework gives a way to construct new pulse sequences and to calculate their suppression order. In Sec.~\ref{sec:set}, we describe the mathematical settings of the dynamical decoupling technique. In Sec.~\ref{sec:CPMG}, we first define our projection pulse sequence and explain its effect as an `atomic' projection. Then we use concatenated projections along different directions to construct more complex pulse sequences with arbitrary suppression order in Sec.~\ref{sec:Concat}. These results comprise the two cornerstones of our theory of CPDD. In Sec.~\ref{sec:CPDD} we formally introduce CPDD by defining the CPDD equivalence classes, which are specified by three integers. In subsection A, a series of properties of CPDD are developed; in subsection B we design a deterministic scheme to construct optimized pulse sequence for each suppression order. A table of known DD schemes is given as well to show how these known schemes fit into our CPDD framework. In Sec.~\ref{sec:DISCUSS}, we discuss why, intuitively, some CPDD sequences are superior than CDD. We also point out a typo in Ref.\cite{Quiroz:13}, which is easily detected within the framework of CPDD. We present our conclusion in Sec.~\ref{sec:conc}.
\section{Dynamical decoupling settings}

\label{sec:set}
We consider a qubit system $S$ coupled to an arbitrary bath  $B$, which forms a closed system on the Hilbert space $\mathcal{H}_S\otimes\mathcal{H}_B$. 
The overall Hamiltonian can be written in the form, \newline
\begin{equation}
H_0 = H_S\otimes\unit_B + \unit_S\otimes H_B +H_{SB},
\end{equation} %
where $\unit$ is the identity operator, $H_S$ is the pure system Hamiltonian, $H_B$ is the pure environment Hamiltonian, and $H_{SB}$ is the interaction between the qubit and the bath. In the following, we will assume the interaction takes the general linear form,\newline 
\begin{equation}
H_{SB} = \sigma_x\otimes B_x + \sigma_y\otimes B_y + \sigma_z\otimes B_z,
\label{eq:HSB}
\end{equation} %
where $B_x,B_y$ and $B_z$ are abitary operator on $\mathcal{H}_B$.

Decoherence can be suppressed by adding a control field solely to the system, i.e., $H_c(t)\otimes\unit_B$. If the control field is a series of pulses, then the suppression effect can be understood as a symmetrizing procedure \cite{Viola:99}.
In the toggling frame, the Hamiltonian is transformed into,\newline
\begin{equation}
\tilde{H}(t) = U_c^\dagger(t) H_0 U_c(t) ,
\end{equation}
where $U_c(t) = \mathcal{T}\exp\lbrace-i\int_0^t H_c(t)dt\rbrace$ is the evolution operator of the control field.
Here $\mathcal{T}$ is time ordering operator. The evolution of the state vector is governed by the evolution operator,\newline
\begin{equation}
\tilde{U}(t) = \mathcal{T}\exp\Big\lbrace-i\int_0^t \tilde{H}(t)dt\Big\rbrace.
\end{equation} %
If the sequence length is $\tau_c$, we can define an average Hamiltonian $\bar{H}(\tau_c)$ by $\tilde{U}(\tau_c) = \exp(-i\bar{H}\tau_c)$. By expanding $\tilde{U}(\tau_c)$, we can collect terms according to different order of $\tau_c$,\newline
\begin{equation}
\mathcal{T}\exp\Big\lbrace-i\int_0^t \tilde{H}(t)dt\Big\rbrace =
e^{-i[\bar{H}^{(0)}+\bar{H}^{(1)}+...]\tau_c}. 
\label{eq:average}
\end{equation} %
The first two terms are
\begin{eqnarray}
&\bar{H}^{(0)}&= \frac{1}{\tau_c}\int_0^{\tau_c}\tilde{H}(t)dt, \notag\\
&\bar{H}^{(1)}&=\frac{-i}{2\tau_c}\int_0^{\tau_c}dt_1\int_0^{t_1}dt_2 [\tilde{H}(t_1),\tilde{H}(t_2)], \notag\\
&...&
\label{eq:magnus}
\end{eqnarray}
We make the separation,\newline 
\begin{equation}
\label{eq:err1}
\bar{H} = \unit_s\otimes H_B + \bar{H}_{\rm{err}} ,
\end{equation}
in which $\bar{H}_{\rm{err}}$ is the part responsible for undesired evolution of the qubit.

In the context of dynamical decoupling, the only measure of the performance of a pulse sequence is the suppression order. The suppression order $N$ is defined through the following equation,
\begin{eqnarray}
\label{eq:err2}
 \bar{H}_{\rm{err}} = \sum_{m=N+1}^{\infty}\bar{H}^{(m)},
\end{eqnarray} %
which means the corresponding pulse sequence achieves $N^{th}$ order decoupling.
\begin{comment}
\begin{mydefinition}
\label{def:supr}
The suppression order N achieved by a pulse sequence A with sequence length $\tau_c$ is the leading order in the Magnus expansion of the average Hamiltonian,
\begin{equation}
\mathcal{T}\exp(-i\int_0^{\tau_c}\tilde{H}(t)dt) = \exp(-i\sum_{m=N}^{\infty}\bar{H}^{(m)}\tau_c).
\end{equation} %
\end{mydefinition}
\end{comment}

\section{Projection pulse sequence}

\label{sec:CPMG}

Now consider that we use $K$ uniform ideal $\pi$ pulses as the control field. Thus the field takes the form $H_c(t) = -\sum_{j=1}^K \frac{\pi}{2}\sigma_j\delta(t-t_j)$, where $t_j = j\tau_d$, $\tau_d$ is the pulse interval and $\sigma_j\in \lbrace \unit,\sigma_x,
\sigma_y,\sigma_z\rbrace$. In the following, we use $P_K...P_1$, $P_K\in\lbrace I,X,Y,Z\rbrace$ to represent corresponding pulse sequence. In the limit of $\tau_d\rightarrow 0$, we have a continuous pulse sequence, and only the zeroth-order term in the Magnus expansion survives,\newline
\begin{equation}
\bar{H}^{(0)}= \frac{1}{\tau_c}\int_0^{\tau_c}U_c^\dagger(t)H_0U_c(t)dt.
\label{eq:1st}
\end{equation} %
Also, since the pulses are ideal, $U_c(t)$ is a piece-wise constant function. The first-order average Hamiltonian reduces to \newline
\begin{equation}
\bar{H}^{(0)}= \frac{1}{K}\sum_{j=1}^K U_{c}^\dagger (t_j) H_0U_{c}(t_j).
\label{eq:symmetrize}
\end{equation} 
If we chose pulses such that $U_{c}(t_j)$ go through each element of certain group such as $\mathcal{G}=\lbrace I,X,Y,Z\rbrace$, Eq.~(\ref{eq:symmetrize}) is just a symmetrizing procedure which projects $H_0$ onto the commutant of the group algebra \cite{Viola:99}.

A more intuitive way to view the effect of pulse sequence is to look at them as combination of basic projections \cite{KhodjastehLidar:04}. The simplest pulse sequence is $P_iP_i$, which is similar to the CPMG pulse sequence \cite{Bloem:07,Carr:1954}. Hereafter we call it the projection pulse sequence and use the notation $p_i = P_iP_i$. To explain its projecting effect, we consider the spin-boson model, which induces the longitudinal decay of the spin. The interaction takes the form,
\begin{equation}
H_{SB} = \sum_k(g_k\sigma^+\otimes b_k+g_k^*\sigma^-\otimes b_k^\dagger),
\end{equation} %
where $b_k$ is the annihilation operator of a photon with momentum $k$, and $g_k$ is the coupling strength between photon with mode $k$ and the spin. Here $\sigma^\pm$ is the creation (annihilation) operator, $\sigma^\pm = \sigma_x\pm i\sigma_y$. 

If we apply $p_z=ZZ$ to the system, then $U_c(t_j)\in \mathcal{G} = \lbrace\unit,Z\rbrace$. Using Eq.~(\ref{eq:symmetrize}), the interaction term will be completely removed in the continuous limit due to the fact $\sigma_z\sigma_{x(y)}\sigma_z = -\sigma_{x(y)}$. Therefore, geometrically the pulse sequence $p_z$ projects the Hamiltonian along the $z$ direction. 

However, in real experimental conditions, there is an upper limit of the pulse switching rate; thus $\tau_d$ is finite. Although the Magnus expansion is still valid, so long as $\tau_c\omega_c << 1$, where $\omega_c$ is the cut-off frequency of the bath, the projection is not exact anymore due to higher-order corrections.  

\begin{mytheorem}
\label{theo:CPMG}
: Assume the interaction between a single qubit and bath takes the form of Eq.~(\ref{eq:HSB}). After applying the projection pulse sequence $p_j$ ($j=x,y,z$) with pulse interval $\tau_d$, the zeroth order error Hamiltonian defined in Eq.~(\ref{eq:err1}) and Eq.~(\ref{eq:err2}) is given by,\newline
\begin{equation}
\bar{H}^{(0)}_{\rm{err}} \equiv \pi_j^{(0)}H_0=\sigma_j\otimes B_j .
\end{equation}
Thus the full error average Hamiltonian is given by,
\begin{equation}
\bar{H}_{\rm{err}} = \sigma_j\otimes [B_j+k_j^{(1)}(\tau_d)]+\sum_{i\perp j}\sigma_i\otimes [B_i^{(1)}+k_{i}^{(2)}o(\tau_d^2)],
\end{equation} %
where we use $\pi_j^{(0)}$ to represent the mapping from $H_0$ to $\bar{H}^{(0)}_{\rm{err}}$ that is induced by the projection pulse sequence $p_j$. The symbol $\perp$ represents the directions orthogonal to direction $j$ and $B_i^{(1)}\sim k_i^{(1)}o(\tau_d)$. Here $k^{(n)}_i$ is some combination of commutators of the bath operators with dimension of $[H]^n$.
\end{mytheorem}
This projection point of view gives a geometrical and intuitive way to understand the effect of the pulse sequence $P_jP_j$. The full proof is given in Appendix A.

In summary, we have shown that the effect of the projection pulse sequence $p_i =P_iP_i$ is to project the Hamiltonian along $i$ direction up to first order.
\section{Concatenation of cyclic pulse sequences as successive projections}

\label{sec:Concat}

The higher-order terms remaining in the $H_{\rm{err}}$, after applying the projections, will coherently add up  with time. To achieve higher-order suppression, we need to project the  Hamiltonian along different directions successfully. We will show in this section that the effect of the concatenation of cyclic pulse sequences is to apply successively the projections induced by each pulse sequence.

\begin{comment}
\begin{mydefinition}
\label{def:cyc}
A pulse sequence is called cyclic when the generated evolution operator  is periodical with period $\tau_c$ up to a phase factor,
\begin{equation}
U_c(n\tau_c) = e^{i\phi}U_c(\tau_c),
\label{eq:cyc}
\end{equation} %
where $n\in\mathcal{N}=0,1,2..$ and $\phi$ is an arbitrary phase.
\end{mydefinition}
\end{comment}
A pulse sequence is called cyclic when the generated evolution operator is periodic with period $\tau_c$ up to a phase factor,
\begin{equation}
U_c(n\tau_c) = e^{i\phi}U_c(\tau_c),
\label{eq:cyc}
\end{equation} %
where $n=0,1,2~...$ and $\phi$ is an arbitrary phase.

An equivalent definition of a cyclic pulse sequence is that the product of all the pulses is equal to one, up to an arbitrary phase,  
\begin{equation}
\prod_{i=1}^{K}P_i = e^{i\phi}\unit,
\label{eq:cyc1}
\end{equation}
which follows directly from Eq.~(\ref{eq:cyc}) when $n=0$.
We will see that this property is necessary for the proof of the equivalence between concatenation and successive projections. 
\begin{comment}
\begin{mycorollary}
\label{def:cyc1}
A cyclic pulse sequence $P_K,P_{K-1},...,P_1$ satisfies,
\begin{equation}
\prod_{i=1}^{K}P_i = e^{i\phi}\unit.
\label{eq:cyc1}
\end{equation}
\end{mycorollary}

\begin{proof}
From [Eq.~(\ref{eq:cyc})], when $n=0$, 
\begin{equation}
U_c(0)=e^{i\phi}U_c(\tau_c).
\notag
\end{equation}
Since $U_c(0)=\unit$ and $U_c(\tau_c) = \prod_{i=1}^{K}P_i$, we have [Eq.~(\ref{eq:cyc1})]
\end{proof}\newline
\end{comment}

Another useful property of the cyclic pulse sequence is that the concatenation of two cyclic pulse sequences is still cyclic. The proof is straightforward and we include it in the Appendix A. 
Having the definition of cyclic pulse sequence, we now prove the second basic theorem of our CPDD scheme.
\begin{mylemma}
\label{lem:conca}
: Consider two pulse sequences A and B, $P_{K}^i...P_{1}^i$, where i = $A,B$, with the same pulse interval. The first pulse sequence $A = P_jP_j$ ($j\in\lbrace x,y,z\rbrace$), which is a projection pulse sequence, and sequence $B$ is concatenated from multiple projection pulse sequences. A third pulse sequence C is constructed by concatenating  A and B, $C = A[B] \equiv P_jBP_jB$. 
\begin{comment}

The transformation of the Hamiltonian induced by each pulse sequence $i$ are given by $\pi_i$, where
\begin{eqnarray}
\pi_i H_{SB}= \frac{1}{K_i}\sum_{m=1}^{K_i}U_{m}^{i\dagger} H_{SB}U_{m}^i, 
\label{eq:projection}
\end{eqnarray}
where $i = A,B,C$ and $U_{m}^i = \prod_{j<=m}P^i_j$.
content...
\end{comment}
The following relationship is true,
\begin{equation}
\pi_C^{0} = \pi_B^{0}\pi_A^{0},
\end{equation}  
where the mapping $\pi_i^{0}$ induced by applying the pulse sequence $i$ is defined in Theorem \ref{theo:CPMG}.
\end{mylemma}
Lemma \ref{lem:conca} is the theoretical cornerstone of this work. It explains why concatenation can increase the suppression order, which is not so obvious. The cyclic properties and the changeability of different pulses (up to an irrelevant phase factor) are necessary for the proof. The full proof is included in the Appendix A. 

\section{Concatenated projections dynamical decoupling}

\label{sec:CPDD}
What really distinguishes our work from the CDD scheme is that we chose the projection pulse sequence as the basic element of concatenation. Motivated by Theorems \ref{theo:CPMG} and \ref{theo:con}, we define our concatenated-projection dynamical decoupling (CPDD) as a new way to construct pulse sequences by applying projections along different directions successively. Since each projection kills the interaction terms orthogonal to it by one more order, by appropriately combining different projections, our CPDD can achieve arbitrarily high suppression order.
\begin{mydefinition}
\label{def:CPDD}
: A CPDD pulse sequence is specified by an ordered series $i_N,i_{N-1},...,i_1$. It is constructed by concatenating N projection pulse sequences successively, $A = p_{i_N}[p_{i_{N-1}}[...[p_{i_1}\underbrace{]...]}_N$, where $i_j\in \lbrace x,y,z \rbrace$ and $1\leq j\leq N$.
\end{mydefinition} 
The suppressing effect of the CPDD sequence on the Hamiltonian follows immediately from the combination of the effects of projections and concatenation. 
\begin{mytheorem}
\label{theo:con}
 : Consider a CPDD pulse sequence A specified by $i_N,i_{N-1},...,i_1$. After applying pulse sequence A, the  average interaction Hamiltonian is given by\newline
\begin{equation}
\bar{H}_{err} = \sum_{i=x,y,z}\sigma_i\otimes[ B_i^{(d_i)} + k_i^{(d_i+1)}o(\tau_d^{d_i+1})],
\end{equation} %
where 
\begin{equation}
d_i = \sum_{j\perp i}n_j,
\end{equation}
and $n_j$ is the number of $p_j$ sequences.
\end{mytheorem}
%\begin{proof}
\textbf{Proof} : Repeatedly using Lemma \ref{lem:conca}, the leading order of the error average Hamiltonian after applying sequence $A$ is given by\newline
\begin{equation}
\pi_A^{(0)} H_0 = \pi^{(0)}_{i_1}\pi_{i_2}^{(0)}...\pi_{K_A}^{(0)}H_{0}~.
\end{equation}
From Theorem \ref{theo:CPMG}, each projection remove the first order term in the perpendicular direction. Therefore,
\begin{equation}
 \bar{H}_{err}^{(0)} = \sum_{i=x,y,z}\sigma_i\otimes B_i^{(d_i)},
\end{equation} %
where 
\begin{equation}
\label{eq:sum}
d_i = \sum_{j\perp i}n_j.
\end{equation}
\newline\newline
From Theorem \ref{theo:con} we can see that the effect of $n_i$ $p_i$s is to suppress the error Hamiltonian along the $i$ direction to $n_i$th order.
From Eq.~(\ref{eq:sum}) we also notice that the order of how different projection pulse sequences is concatenated does not affect the leading order of the average Hamiltonian along each direction. Therefore we can define an equivalence relationship between different CPDD pulse sequences,
\begin{mydefinition}
\label{def:eqv}
 : Consider two pulse sequences $A$ and $A'$. The leading order of the error Hamiltonians induced by each of them are $\bar{H}^{(0)}_{err} = \sum_{i=x,y,z}\sigma_i\otimes B_i^{(d_i)}$ and $\bar{H}^{(0)'}_{err} = \sum_{i=x,y,z}\sigma_i\otimes B_i^{(d_i')}$. We define $A$ and $A'$ to be equivalent to each other\newline
\begin{equation}
A\sim A',
\end{equation}
if the leading order of the average error Hamiltonians are the same along each direction, namely $n_i=n_i'$.
\end{mydefinition}
It can be easily proved that the relationship defined above satisfies the three properties of an equivalence relationship. 

Therefore, for a CPDD sequence specified by sequence $a = i_N,i_{N-1},...,i_1$, all CPDD sequences specified by $a$'s permutations $A\lbrace i_N,i_{N-1},...,i_1\rbrace$ form an equivalence class. By the virtue of the equivalence class, only three numbers $n_x,n_y,n_z$ are needed 
to completely specify a CPDD class.

\begin{mydefinition}
: A CPDD class is defined as an equivalence class with equivalence relationship defined in Definition \ref{def:eqv}, specified by three integers, $\lbrace n_x,n_y,n_z\rbrace$. The structure of the pulse sequence can be generated by concatenating all $n_i$ $p_i$ sequences ($i=x,y,z$) in arbitrary order.
\end{mydefinition}
From the definition of CPDD and Theorem \ref{theo:con}, we derive a series of properties satisfied by CPDD sequences and their equivalence classes.
\subsection{Properties of CPDD }

\label{subsec:propert}

Due to the way concatenation connects two pulse sequences, we can derive two properties that the structure of each CPDD sequence must satisfy.\\\\
1. For an arbitrary CPDD pulse sequence, each odd site has the same kind of $\pi$ pulse.

%\begin{mytheorem}
%\label{theo:odd}
%For an arbitrary CPDD pulse sequence, %$P_K,P_{K-1},...P_2,P_1$, \newline
%\begin{equation}
%P_{2m+1} = P_0, ~~~m=0,1,...
%\end{equation}
%where $P_0\in\lbrace \unit,P_x,P_y,P_z\rbrace$.
%\end{mytheorem}
We prove this by induction.
Consider a pulse sequence $A = P_KP_{K-1}...P_{1}$, which is concatenated from $N$ projection pulse sequences. Pulse sequences $ A_n$ are concatenated from the first $n$ ($1\leq n\leq N$) of them. The following relations are satisfied,
\begin{align}
\label{eq:induct}
A_n = p_i[A_{n-1}],
\end{align}  
where $p_i$ is the $n$th projection pulse sequence. Assume for subsequence $A_{n-1}$ the pulses are the same for each odd site,
\begin{equation}
A_{n-1} = P_{2n-2}P_0.....P_2P_0.
\label{eq:recur}
\end{equation}
Let's examine the pulse sequence $A_n$,
\begin{eqnarray}
A_n &=& p_i[A_{n-1}]\notag\\
&=&(P_iP_{2n-2})P_0...P_2P_0(P_iP_{2n-2})P_0...P_2P_0.
\end{eqnarray}
Therefore, the fact that all the pulses at each odd site are the same still holds. 

Since $A_2 = p_i[p_j]=(P_iP_j)P_j(P_iP_j)P_j$, which also has the same kind of pulse on its odd sites, by induction we have proved that for each odd sites of $A_N$ the pulses are the same:
\begin{equation}
P_{2m+1} = P_0, ~~~m=0,1,...
\end{equation}\\\\
2. For an arbitrary CPDD pulse sequence, the first half and the second half subsequences are the same.\\
Again this property follows from the definition of concatenation. Using Eq.~(\ref{eq:induct}) for $n = N$ we have
%\begin{equation}
%P_m = P_{m-K/2}, ~~~m = 1,2,...,K.
%\end{equation}
\begin{eqnarray}
A = p_{i_n}[A_{n-1}].
\end{eqnarray} 
Assume $A_{n-1} = P_{2N-2}P_{2N-3}...P_1$, we have
\begin{equation}
A = (P_{i_n}P_{2N-2})P_{2N-3}...P_1(P_{i_n}P_{2N-2})P_{2N-3}...P_1.
\end{equation}
Obviously sequence $A$ is composed the same two copies of $A_{N-1}$, or more precisely
\begin{equation}
P_m = P_{m-N/2}, ~~m=1,2,...,N.
\end{equation}\\\\
\begin{comment}
content...
\begin{mytheorem}
\label{theo:K}
For a CPDD pulse sequence specified by the number of projections along each direction, $\lbrace n_x,n_y,n_z \rbrace$, the number of pulses or the sequence length K is given by,\newline
\begin{equation}
K_{n_x,n_y,n_z} = \prod\limits_{i = x,y,z}2^{n_i}.
\label{eq:length}
\end{equation} % 
\end{mytheorem}
\end{comment}
3. For CPDD class $\lbrace n_x,n_y,n_z\rbrace$, the number of pulses or sequence length $K$ is given by
\begin{equation}
K_{n_x,n_y,n_z} = 2^{n_x+n_y+n_z}.
\label{eq:length}
\end{equation} % 
\\
The proof is straightforward. At first, each basic projection is induced by the two same pulses. Secondly, the length of two concatenated pulse sequences, $A$ and $B$, is equal to the product of the length of each two,
$K_{A[B]} = K_AK_B$. Therefore, for a pulse sequence composed of $n_i$ pairs of ($P_i,P_i$), the total pulse number $K$ is given by Eq.~(\ref{eq:length}).\\\\
\begin{comment}
\begin{mytheorem}
\label{theo:suppr}
For a CPDD class specified by the number of projections along each direction, $\lbrace n_x,n_y,n_z \rbrace$, the corresponding suppression order N, which is defined in Definition \ref{def:supr},
is given by,\newline
\begin{equation}
N_{n_x,n_y,n_z} = \min{\lbrace n_y+n_z,n_x+n_z,n_x+n_y\rbrace}.
\label{eq:suppr}
\end{equation} % 
\end{mytheorem}

\end{comment}
4. For the CPDD class $\lbrace n_x,n_y,n_z\rbrace$, the suppression order achieved is given by,\newline
\begin{equation}
N_{n_x,n_y,n_z} = \min{\lbrace n_y+n_z,n_x+n_z,n_x+n_y\rbrace}.
\label{eq:suppr}
\end{equation} % 

From Theorem \ref{theo:con}, the leading order of the error Hamiltonian induced by any pulse sequence in the CPDD class $\lbrace n_x,n_y,n_z \rbrace$ is given by
\begin{equation}
\bar{H}^{(0)}_{err} = \sum_{i=x,y,z}\sigma_i\otimes B_i^{(d_i)},
\end{equation}
where $d_i = \sum_{j\perp i}n_j$. Since the suppression order $N$ is defined as the leading order of $\bar{H}$, $N=\min{\lbrace d_x,d_y,d_z\rbrace}$.
\\\\
From the expression of suppression order, Eq.~(\ref{eq:suppr}), we can see that simply increasing pulse numbers (number of projections) does not necessarily increase the suppression order. Actually, in the framework of CPDD, only for pulse sequences with certain pulse numbers can the suppression order increase.\\\\
\begin{comment}
\begin{mytheorem}
\label{th:5}
For given suppression order $N$, the minimum number of pulses, $K_{min}$, required to achieve such suppression order is given by \newline 
\begin{equation}
\log_2K_{inc} = \frac{1}{2}[3N+\frac{1}{2}(1-1^{\oplus N})],
\label{eq:Kinc}
\end{equation} %{
where $1^{\oplus N} = \underbrace{1\oplus 1 \oplus ~...~\oplus 1}_{k}$.
\end{mytheorem}
\end{comment}
5.~For a given suppression order $N$, the minimum number of pulses, $K_{\rm{min}}$, required to achieve such suppression order is given by \newline 
\begin{equation}
\log_2(K_{\rm{min}}) = \frac{1}{2}[3N+\frac{1}{2}(1-1^{\oplus N})],
\label{eq:Kinc}
\end{equation} %{
where $1^{\oplus N} = \underbrace{1\oplus 1 \oplus ~...~\oplus 1}_{N}$.\\\\
To achieve suppression order $N$, $3N$ terms in the interaction Hamiltonian need to be eliminated due to the form of $H_{SB}$, Eq.~(\ref{eq:HSB}). However, each basic projection $\pi_i$ requires two pulses, which implies that only an even number of terms can be eliminated for a given CPDD sequence. Therefore, we need to add one more pulse depending on whether $N$ is odd or not. Dividing the total number of eliminated terms by two we have, \newline
\begin{equation}
\sum_{i=x,y,z}n_i =  \frac{1}{2}[3N+\frac{1}{2}(1+(-1)^{N+1}].
\end{equation} %
Using the results of property 3, we have Eq.~(\ref{eq:Kinc}).
The mysterious series $4,8,32,64,256...$ was first found by a genetic algorithm in Ref.~\cite{Quiroz:13} which is now understood within the framework of our unifying CPDD.
\subsection{Optimized uniform dynamical decoupling}

The pulse sequences corresponding to the pulse number in Eq.~(\ref{eq:Kinc}) uses the minimum number of pulses at each suppression order. This optimized uniform dynamical decoupling (OUDD) scheme can be represented using the CPDD indexes as
\begin{equation}
\label{eq:OUDD}
\mbox{OUDD}_k :\frac{1}{2}\lbrace k-1^{\oplus k},k+1^{\oplus k},k+1^{\oplus k}\rbrace.
\end{equation} %
The suppression order of OUDD$_k$ is $N_k = k$ and the sequence length is given by $K_{\rm{min}}$ in Eq.~(\ref{eq:Kinc}).\\\\
As we can see, some of the pulse sequences are particular levels of CDD$_l$ and GA8$_l$. To compare with OUDD and other known DD schemes, we list the corresponding CPDD indexes of known DD schemes in Table \ref{tab:1} below.\newline
\begin{minipage}{\linewidth}
\centering
\captionof{table}{Known DD schemes represented as CPDD } \label{tab:1} 
\begin{tabular}{c|c|c|c|c}
\hline\hline
$\lbrace n_x,n_y,n_z\rbrace$~ & ~ Name ~& ~Pulse sequence ~&~ K ~& ~N~\\
\hline
$\lbrace 0,0,1\rbrace$ & Projection & $P_iP_i$ & 2 & 0\\
\hline
$\lbrace 0,1,1\rbrace$ & PDD(CDD$_{1}$) & $P_iP_jP_iP_j$ & 4 & 1\\
\hline
$\lbrace 1,1,1\rbrace$ & GA8$_a$ & $IP_iP_jP_iIP_iP_jP_i$ & 8 & 2\\
\hline
$\lbrace 0,l,l\rbrace$ & CDD$_{l}$ & CDD[CDD$_{l-1}$] & $4^l$ & $l$\\
\hline
$\lbrace l,l,l\rbrace$ & GA8$_{l}$ & GA8$_a$[GA8$_{l-1}$] & $8^l$ & $2l$\\
\hline\hline
\end{tabular}
\bigskip

\end{minipage}

\section{Discussion}

\label{sec:DISCUSS}

Although CDD also relies on concatenation, the fact that our CPDD uses basic projections as building blocks makes finding more efficient pulse sequences possible. To make this clear, we consider the CDD from the view point of our CPDD. CDD$_1$ can be considered as the concatenation of two different projections,  CDD$_1 = X(YY)X(YY) = ZYZY$ \cite{KhodjastehLidar:07}. Therefore the effect of CDD$_1$ is to apply projections along the y and x direction successively,\newline
\begin{eqnarray}
 \pi_{CDD_1}^{(0)}H_{SB} &\equiv& \pi^{(0)}_y\pi_x^{(0)}H_0 \notag\\
&=&\pi_y^{(0)}[\sigma_x\otimes B_x+\sigma_y\otimes B_y^{(1)}+\sigma_z\otimes B_z^{(1)}]\notag\\
&=&\sigma_x\otimes B_x^{(1)}+\sigma_y\otimes B_y^{(1)}+\sigma_z\otimes B_z^{(2)} \notag\\
&\sim& k^{(1)}o(\tau_d^1).
\end{eqnarray} %
As we can see, CDD$_1$ completely removes the zeroth order interaction terms, thus achieving suppression order $N_{\rm{CDD}_1} = 1$.

Now consider CDD$_2$, which is the concatenation of two CDD$_1$ sequences: we write explicitly the process of the successive projections ,\newline 
\begin{eqnarray}
\pi^{0}_{CDD_2} H_{SB} &=& \pi_y^{(0)}\pi_x^{(0)} \pi_y^{(0)}\pi_x^{(0)} H_0 \notag\\
&=&\pi_y^{(0)}\pi_x^{(0)}[\sigma_x\otimes B_x^{(1)}+\sigma_y\otimes B_y^{(1)}+\sigma_z\otimes B_z^{(2)}]\notag\\
&=&\sigma_x\otimes B_x^{(2)}+\sigma_y\otimes B_y^{(2)}+\sigma_z\otimes B_z^{(4)} \notag\\
&\sim& k^{(2)}o(\tau_d^2).
\end{eqnarray}
As we can see from above, a total of eight eliminations (each projection pulse sequence eliminates two terms in the orthogonal directions) are used to completely remove the first two orders of the interaction $H_{SB}$. However, the two additional eliminations of $B_z$ does not contribute to further increasing of the suppression order. To avoid this, we consider projecting along each direction exactly once, namely $\pi_x\pi_y\pi_z$, which belong to the CPDD class $\lbrace 1,1,1\rbrace$. Translating the projection back to corresponding pulse sequence according to the rule of concatenation, we have\newline
\begin{eqnarray}
\pi_x^{(0)}\pi_y^{(0)}\pi_z^{(0)} &:& p_x[p_y[p_z]] \notag\\
&=& p_x[Y(ZZ)Y(ZZ)] \notag\\
&=& X(XZXZ)X(XZXZ) \notag\\
&=& IZXZIZXZ ,
\end{eqnarray} %
which only uses eight pulses and six projections. This was first found by a genetic algorithm and called GA8$_a$ in Ref.~\cite{Quiroz:13}. \\
\\
To achieve suppression order of $N = 2$, CDD$_2$ needs $16$ pulses while GA8a sequences only requires $8$ pulses.  This efficiency of using pulses comes from the very fact that GA8a uses basic projections $\pi_i$ as building blocks while CDD$_2$ uses composed projections $\pi_i\pi_j$ ($i\neq j$) as building blocks.\\
\\
In Ref.~\cite{Quiroz:13}, the author claimed to find another 8 pulse sequence GA8$_b$=$Z(XYXY)Z(XYXY)$ which also achieved suppression order of $N = 2$. From the structure of GA8$_b$, we know the projections induced by it is,
\begin{equation}
\pi^{(0)}_{GA8_b}=\pi_z^{(0)}\pi_z^{(0)}\pi_y^{(0)},
\end{equation}
which belong to CPDD class $\lbrace 0,1,2\rbrace$. Using the results from Eq.~(\ref{eq:suppr}), the suppression order of GA8$_b$ is equal to 1. Therefore the claim in Ref.~\cite{Quiroz:13} is typo.\\
\\
To double check our results, we also use the multi-precision package {\it mpmath} \cite{mpmath} to compute the suppression order of both GA8$_a$ and GA8$_b$ for a 5-spin model with random coupling constants. Here the distance $D$ is defined as the distance between an actual evolution operator and the unit operator \cite{Quiroz:13},
\begin{equation}
D(U,\unit_S) = \sqrt{1-\frac{1}{d_\mathcal{H}}||\Gamma||_{Tr}},
\end{equation}
where $d_\mathcal{H}$ is the dimension of the Hilbert space, $\Gamma = \tr_S[U]$ and $||\cdot||_{Tr}$ represents the trace norm. An upper bound of $D$ can be calculated\cite{Quiroz:13}, 
\begin{equation}
D\lesssim O[\tau_d^{N+1}].
\end{equation}
Therefore, we can extract the suppression order by plotting $D$ versus $\tau_d$ in the log-log diagram. As we can see in Figure \ref{fig:GA8}, GA$_{8a}$ achieves higher suppression order than GA$_{8b}$ which is in agreement with the argument from the view point of our CPDD.
\begin{figure}
\includegraphics[trim={0 0 0 12cm},clip,width=0.55\textwidth]{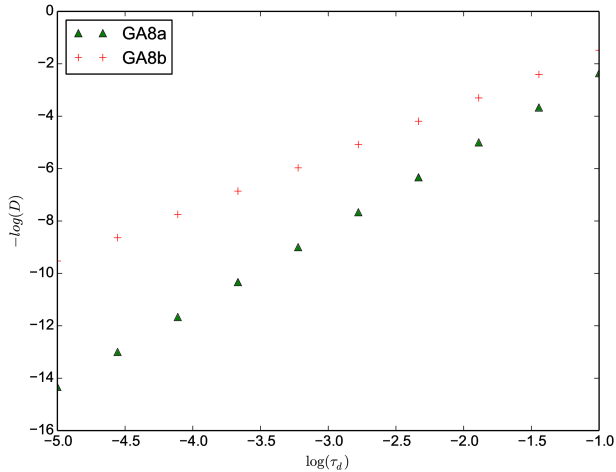}
\caption{\small Comparing the suppression order of GA8$_a$ and GA8$_b$. $D$ is the distance between actual evolution operator and the unit operator, defined in \cite{Quiroz:13}, and $\tau_d$ is the pulse interval. The suppression order $N$ is defined by the relation $D\sim O(\tau_d^{N+1})$. We consider the parameters $J$ in the range of  $J\tau_d\in[10^{-6},10^{-1}]$, where $J$ is the norm of the interaction Hamiltonian.}\label{fig:GA8}
\end{figure}
\newline
\section{Conclusion}
\label{sec:conc}
% % % % % % % % % % % % % % % % % % % % % % % % % %
We have developed tour CPDD schemes in which CPDD sequences are concatenated from different projection pulse sequences. We also define CPDD equivalence classes as the set of pulse sequences that have the same leading order along each direction in the Magnus expansion of the average Hamiltonian. Based upon the definition of our CPDD pulse sequences, we prove a series of properties about the structure of the sequences that must hold for CPDD. We also give a formula to calculate suppression order for a given CPDD class specified by $\lbrace n_x,n_y,n_z\rbrace$. We propose the optimal uniform DD sequence given in Eq.~(\ref{eq:OUDD}) use in experiments, since each of the sequences in the series achieves its suppression order using minimum pulse numbers. Although some of OUDD are already known DD sequences, our CPDD framework gives a unifying and consistent way to both understand and construct all of them.\\
\\
The main advantage of using UDD is that the pulse number needed scales linearly with the suppression order, $N_{UDD}\sim \mathcal{O}(K)$, which is much more efficient than the exponentially dependence in CPDD. However, UDD is subject to several difficulties. Firstly it is valid only for environmental spectrum with a hard cut-off \cite{Pasini:2010,Uhrig:08}, and secondly it is very sensitive to pulse errors \cite{Xiao:11,Lange:2010,Ryan:2010,Souza:2012}. However, for our CPDD, especially the OUDD class, some pulse sequences have rotation symmetry thus making them robust against pulse error. Although we have not given a rigorous bound analysis for CPDD here, the results and the calculation should be similar to those in Ref.~\cite{KhodjastehLidar:07}: the distance between the actual state and the desired state goes to zero as the concatenation level goes to infinity. More importantly, pulse errors are suppressed along with series of concatenations as long as the error is not too large.\\
\\
Moreover, although the suppression order of CPDD is still exponentially dependent on the pulse number, $N_{\rm{CPDD}}\sim\mathcal{O}(\log_2K)$,  the reduction of the number of pulses for the same suppression order is exponentially large compared to the case of CDD, $N_{\rm{CDD}}\sim\mathcal{O}(\log_4K)$.\\
\\
Given the structure of our CPDD, one can see that two directions to try to find more efficient pulse sequences are worth exploring. Firstly, to find basic sequences, which could achieve higher $N$ to $K$ ratio than projection pulse sequences. The concatenated UDD \cite{Uhrig:09,Lidar:10} obviously belongs to this class. However, the applicability of Lemma 1 in the context of non-uniform DD is questionable and requires a new proof. Secondly, we seek to find new ways to combine two pulse sequences or two projections. Since the fact that exponentially large pulse numbers are needed to achieve high suppression order resulting from the concatenation , new ideas in this direction may greatly improve the efficiency of uniform dynamical decoupling.

\acknowledgments We thank for helpful discussion with Robert Lanning and Zhihao Xiao%
. This work is supported by the Air Force Office of Scientific Research, the Army Research Office, and the National Science Foundation. 

\appendix

\label{append}
\section{}
In this appendix we include the proof of Theorem \ref{theo:CPMG} and Lemma \ref{lem:conca}, which are the foundations of our CPDD scheme. A proof about a property of cyclic pulse sequences is also included.\\
\\
\noindent\textbf{Proof of Theorem \ref{theo:CPMG}:}
Without loss of generality, we consider the pulse sequence is $ZZ$, where $Z = -i\sigma_z$. The transformed Hamiltonian $\tilde{H}(t)$ is given by\newline
\begin{equation}
\tilde{H}(t)\equiv\begin{cases}
 H_1, ~0<t<\tau_d\\
 H_2, ~\tau_d<t<2\tau_d~,
\end{cases}
\label{theo1:eq:H}
\end{equation}
where
\begin{eqnarray}
H_1& =& \unit\otimes B_0+\sigma_x\otimes B_x+\sigma_y\otimes B_y+\sigma_z\otimes B_z,~~~~\\
H_2& =& \unit\otimes B_0-\sigma_x\otimes B_x-\sigma_y\otimes B_y+\sigma_z\otimes B_z.~
\end{eqnarray}
After applying the pulse sequence $P_jP_j$, the first and second order expansion of the average Hamiltonian is given by Eq.~(\ref{eq:magnus}),
\begin{eqnarray}
\bar{H}^{(0)}&=&\frac{1}{2}\sum_{i=1}^2 H_i,\\
\bar{H}^{(1)}&=&\frac{-i}{2}\tau_c\sum_{i<j} [H_i,H_j].
\end{eqnarray} %
Using Eq.~\ref{theo1:eq:H}, we have 
\begin{eqnarray}
\bar{H}^{(0)} &=& \unit\otimes B_0 +\sigma_z\otimes B_z,\\
\bar{H}^{(1)} &=& \tau_d[B_0,B_x]\otimes\sigma_x\notag\\ & & + \tau_d[B_0,B_y]\otimes\sigma_y + 2\tau_d[B_y,B_x]\otimes \sigma_z.
\end{eqnarray}
Since the commutator between bath operators is not zero in general, we have 
\begin{equation}
\bar{H}_{err} = \sigma_z\otimes [B_z+k^{(1)}_zo(\tau_d)]+\sum_{i\perp z}\sigma_i\otimes [B_i^{(1)}+k^{(2)}_io(\tau_d^2)],
\end{equation} %
where $B_i^{(1)} = \tau_d[B_0,B_i]$. The same calculation gives similar results for $j = x, y$. Q.E.D.
\\

% % % % % % % % % % % % % % % % % % % % % % % % % %
\noindent\textbf{Proof of Lemma \ref{lem:conca} :}
The concatenated sequence $C$ is given by
\begin{equation}
(P_iP^B_K)P^B_{K-1}...P^B_1(P_iP^B_K)...P^B_1,
\end{equation} 
where the bracket means the there is no free evolution in between the two pulses inside the bracket. Since the projection pulse sequence is cyclic by definition, and sequence $B$ is also cyclic by Theorem \ref{theo:cyc}, 
\begin{equation}
\label{eq:prf1:cyc}
\prod_{i=1}^{K}P^B_i = e^{i\phi}\unit.
\end{equation} %
The $2K_B$ evolution operator of the control field, $U_m^C$ is given by, 
\begin{equation}
U_m^C = \prod_{j\leq m}P^C_j
\end{equation} 
To construct $\pi_A^{(0)}$ and $\pi_B^{(0)}$, we group $U_m^C$ and $U_{m+K_B}^C$ together($m\leq K_B$). \newline
For $1\leq m < K_B$,
\begin{eqnarray}
U_m^C &=&\prod_{j<=m}P^B_j=U_m^B, 
\end{eqnarray}
and,
\begin{eqnarray}
U_{m+K_B}^C &=&\Big(\prod_{j<=m}P^B_j\Big) P_i\Big(\prod_{j>m}^{K_B}P^B_j\Big)\Big(\prod_{j<=m}P^B_j\Big)\notag\\
&=&e^{i\phi}P_i\prod_{j<=m}P^B_j\notag\\
&=& e^{i\phi}P_i U^B_m~,
\end{eqnarray}
where we have used the commutativity of Pauli matrices and the cyclic property Eq.~(\ref{eq:prf1:cyc}) of pulse sequence $B$.

Now adding the action of $U_m^C$ and $U_{m+K_B}^C$ on $H_{SB}$ together, we have
\begin{eqnarray}
\label{eq:prf1:1}
& & U_m^{C\dagger}H_{SB}U_m^C + U_{m+K_B}^{C\dagger}H_{SB}U_{m+K_B}^C \notag\\
&=& U_m^{B\dagger}H_{SB}U_m^B +U_m^{B\dagger}P_i^\dagger H_{SB}P_iU_m^B \notag\\
&=&U_m^{B\dagger}(2\pi_A^{(0)}H_{SB})U_m^B
\end{eqnarray}
If $m = K_B$, we have $U_{K_B}^C$ and $U_{2K_B}^C$, which are 
\begin{eqnarray}
\label{eq:prf1:2}
U_{K_B}^C&=&P_iP_{K}^B\prod_{j<K_B}P^B_j\notag\\
&=&P_i\prod_{j<=K_B}P^B_j\notag\notag\\
&=&e^{i\phi}P_iU_{K_B}^B
\end{eqnarray}
\begin{eqnarray}
\label{eq:prf1:3}
U_{2K_B}^C&=&P_i\Big(\prod_{j<=K_B}P^B_j\Big)P_i\Big(\prod_{j<=K_B}P^B_j\Big)\notag\\
&=&e^{i\phi}U_{K_B}^B.
\end{eqnarray}
Now using Eqs.~(\ref{eq:prf1:1}, \ref{eq:prf1:2}, \ref{eq:prf1:3}), the first order of average Hamiltonian after applying sequence $C$ is given by
\begin{eqnarray}
\bar{H}^{(0)}_C&=& \frac{1}{2K_B}\sum_{m=1}^{2K_B}U_{m}^{C\dagger} H_{SB}U_{m}^C
\notag\\ 
&=&\frac{1}{2K_B} \sum_{m=1}^{K_B}\Big(U_m^{C\dagger}H_{SB}U_m^C + U_{m+K_B}^{C\dagger}H_{SB}U_{m+K_B}^C\Big)\notag\\
&=&\frac{1}{K_B}\sum_{m=1}^{K_B}U_m^{B\dagger}\Big(\frac{1}{2}\sum_{l=1}^{K_A}U_l^{A\dagger} H_{SB}U_l^A\Big)U_m^B\notag\\
&=&\pi_B^{(0)}\pi_A^{(0)}H_{SB}
\end{eqnarray}
Q.E.D.
\begin{mytheorem}
\label{theo:cyc}
The concatenation of two cyclic pulse sequences is still cyclic.
\end{mytheorem}
\textbf{Proof of Theorem \protect\ref{theo:cyc} :}
Consider two cyclic pulse sequences $A$ and $B$. From [Eq.~(\ref{eq:cyc1})] we have
\begin{eqnarray}
\prod_{i=1}^{K_A}P_i^A = e^{i\phi_A}\unit,\\
\prod_{i=1}^{K_B}P_i^B = e^{i\phi_B}\unit.
\end{eqnarray}
The pulse sequence $C$ is constructed by concatenating $A$ and $B$, thus
\begin{eqnarray}
C &=& A[B]\notag\\ 
&\equiv& P_1^A(P_1^B...P_{K_B}^B)P_2^A(P_1^B...P_{K_B}^B)...P_{K_A}^A(P_1^B...P_{K_B}^B)\notag\\
\end{eqnarray} 
Therefore the product of all pulses of sequence C is
\begin{eqnarray}
\prod_{i=1}^{K_C}P_i^C &=& P_1^A (\prod_{i=1}^{K_B}P_i^B)....P_{K_A}^A(\prod_{i=1}^{K_B}P_i^B)\notag\\
&=&P_1^Ae^{i\phi_B}...P_{K_A}^Ae^{i\phi_B}\notag\\
&=&e^{iK_B\phi_B}\prod_{i=1}^{K_A}P_i^A\notag\\
&=&e^{i(K_B\phi_B+\phi_A)}\unit\notag\\
&=&e^{i\phi_C}\unit
\end{eqnarray}
where we define $\phi
_c=K_B\phi_B+\phi_A$. Therefore, the pulse sequence $C$=$A$[$B$] is also cyclic. Q.E.D.
\clearpage

\end{document}